\documentclass{sig-alternate}

\usepackage{amsthm,amssymb,amsmath}
\usepackage{graphics}
\usepackage{bbm}
\usepackage{tikz}
\usepackage[plainpages=false,pdfpagelabels,colorlinks=true,citecolor=blue,hypertexnames=false]{hyperref}
\usepackage{color}

\newtheorem{theorem}{Theorem}

\newtheorem{lemma}[theorem]{Lemma}

\newtheorem{defi}[theorem]{Definition}
\newtheorem{example}[theorem]{Example}

\def\qed{\quad\rule{1ex}{1ex}}

\let\set\mathbbm
\def\lc{\operatorname{lc}}

\begin{document}

\title{Desingularization~Explains Order-Degree~Curves for~Ore~Operators}

\numberofauthors{4}

\author{%
 \alignauthor
 \leavevmode
 \mathstrut Shaoshi Chen\titlenote{Supported by the National Science Foundation (NSF) grant CCF-1017217.}\\[\smallskipamount]
  \affaddr{\leavevmode\mathstrut Dept. of Mathematics / NCSU}\\
  \affaddr{\leavevmode\mathstrut Raleigh, NC 27695, USA}
 \and \mathstrut Maximilian Jaroschek\titlenote{Supported by the Austrian Science Fund (FWF) grant Y464-N18.}\\[\smallskipamount]
  \affaddr{\mathstrut RISC / Joh. Kepler University}\\
  \affaddr{\mathstrut 4040 Linz, Austria}
 \and
 \mathstrut Manuel Kauers\raisebox{.43em}{\normalsize$^{\textstyle\dagger}$}\\[\smallskipamount]
  \affaddr{\mathstrut RISC / Joh. Kepler University}\\
  \affaddr{\mathstrut 4040 Linz, Austria}
 \and
 \mathstrut Michael F. Singer\raisebox{.43em}{\normalsize$^{\textstyle\ast}$}\\[\smallskipamount]
  \affaddr{\mathstrut Dept. of Mathematics / NCSU}\\
  \affaddr{\mathstrut Raleigh, NC 27695, USA}\\[\smallskipamount]
}

\maketitle
\begin{abstract}
    Desingularization is the problem of finding a left multiple of a given Ore
    operator in which some factor of the leading coefficient of the original
    operator is removed. 
    An order-degree curve for a given Ore operator is a curve in the $(r,d)$-plane 
    such that for all points $(r,d)$ above this curve, there exists a left 
    multiple of order~$r$ and degree~$d$ of the given operator.
    We give a new proof of a desingularization result by Abramov and van Hoeij
    for the shift case, and show how desingularization implies order-degree curves
    which are extremely accurate in examples. 
\end{abstract}

%\kern-\medskipamount

\category{I.1.2}{Computing Methodologies}{Symbolic and Algebraic Manipulation}[Algorithms]

%\kern-\medskipamount

\terms{Algorithms}

%\kern-\medskipamount

\keywords{Ore Operators, Singular Points}

%\kern-\medskipamount

\section{Introduction}

We consider linear operators of the form
\[
  L = \ell_0 + \ell_1\partial + \cdots + \ell_r\partial^r,
\]
where $\ell_0,\dots,\ell_r$ are polynomials or rational functions in~$x$, and $\partial$ denotes, for
instance, the derivation $\frac d{dx}$ or the shift operator $x\mapsto x+1$.
(Formal definitions are given later.)
Operators act in a natural way on functions. They are used in computer
algebra to represent the functions~$f$ which they annihilate, i.e., $L\cdot
f=0$.

Multiplication of operators is defined in such a way that the product of 
two operators acts on a function like the two operators one after the other:
$(PL)\cdot f = P\cdot(L\cdot f)$. Therefore, if $L$ is an annihilating operator
for some function~$f$, and if $P$ is any other operator, then $PL$ is also
an annihilating operator for~$f$.

%// contribution 1: new explanation and slight generalization 
%   of desingularization results of abramov/vanhoeij in the shift case
%=> motivation: extend sequences beyond singular points; avoid numerical instability

We are interested in turning a given operator $L$ into a ``nicer'' one by
multiplying it from the left by a suitable~$P$, for two different flavors of
``nice''.  First, we consider the problem of removing factors from the leading
coefficient $\ell_r$ of~$L$. This is known as desingularization and it is needed
for computing the values of $f$ at the roots of $\ell_r$ (provided it is defined
there). Desingularization of differential operators is classical~\cite{ince26}, and
for difference operators, Abramov and van Hoeij~\cite{abramov99b,abramov06b} give an algorithm
for doing it. We give below a new proof of (a slightly generalized version of) 
their results. 

%// contribution 2: a formula for the curve which is [almost] tight
% => motivation: a priori knowledge about the sizes of operators helps to design faster algorithms (e.g., ...)

Secondly, we consider the problem of producing left multiples with polynomial
coefficients of low degree. Unlike the situation for commutative polynomials, a left multiple
$PL$ of $L$ may have polynomial coefficients even if $P$ has rational function
coefficients with nontrivial denominators and the polynomial coefficients of $L$ have
no common factors. In such situations, it may happen that the degrees of the
polynomial coefficients in $PL$ are strictly less than those in~$L$. This
phenomenon can be exploited in the design of fast algorithms because a small
increase of the order can allow for a large decrease in degree and therefore
yield a smaller total size of the operator (``trading order for
degree''). Degree estimates supporting this technique have been recently given
for a number of different computational problems~\cite{bostan07,chen12b,chen12c,bostan12b}. Although
limited to special situations, these estimates can overshoot by quite a lot. Below
we derive a general estimate for the relation between orders and degrees of left
multiples of a given operator~$L$ from the results about desingularization. This
estimate is independent of the context from which the operator~$L$ arose, and it
is fairly accurate in examples.

\section{Overview}\label{sec:2}

% // detailed example for the differential case, with explanations heading towards the main argument.

% annihilator of ((3+x)Exp[x]+(x-1)(x-2)Exp[-x])/(x+1)

Before discussing the general case, let us illustrate the concepts of desingularization and 
trading order for degree on a concrete example. 
Consider the differential operator
\begin{alignat*}1
& L = -(45 + 25 x - 35 x^2 - x^3 + 2 x^4)\\
&\quad{}+ 2(33 - 9 x - 3 x^2 - x^3) \partial\\
&\qquad{}(1 + x)(23 - 20 x - x^2 + 2 x^3) \partial^2\in\set Q[x][\partial],
\end{alignat*}
where $\partial=\frac d{dx}$. 
That $L$ is desingularizable at (a root of) $p:=23 - 20 x - x^2 + 2
x^3$ means that there is some other operator $P\in\set Q(x)[\partial]$ 
such that $PL$ has coefficients in $\set Q[x]$ and its leading coefficient no longer
contains~$p$ as factor. Such a $P$ is called a
desingulariz\emph{ing} operator for $L$ at~$p$ and $PL$ the
corresponding desingulariz\emph{ed} operator. In our example,
\begin{alignat*}1
 P = \frac{299}{p} \partial+\frac{1035-104 x-136 x^2}p\in\set Q(x)[\partial]
\end{alignat*}
is a desingularizing operator for $L$ at~$p$, the desingularized operator is
\begin{alignat*}1
&PL=(-2350-2055 x+104 x^2+136 x^3)\\
&\quad{}+(2151+281 x+136 x^2)\partial\\
&\qquad{}+ (1932+931 x-240 x^2-136 x^3)\partial^2 + 299 (1+x)\partial^3.
\end{alignat*}
A desingularizing operator need not exist. For example, it is impossible to
remove the factor $x+1$ from the leading coefficient of $L$ by means of
desingularization. In Section~\ref{sec:desing} we explain how to check for a
given operator $L$ and a factor $p$ of its leading coefficient whether a
desingularizing operator exists, and if so, how to compute it.

Desingularization causes a degree drop in the leading coefficient but may affect
the other coefficients of the operator in an arbitrary fashion. However, a
desingularizing operator can be turned into an operator which lowers the degrees
of all the coefficients. To this end, multiply $P$ from the left by some
polynomial $q\in\set Q[x]$ for which the coefficients of $pqP$ have low degree
modulo~$p$, i.e., for which $qP=\frac1pP_1+P_2$ where $P_1,P_2\in\set
Q[x][\partial]$ and $P_1$ has low degree coefficients. In our example, a good choice is
$q=(-43+34x)/299$, i.e.
\[
  P_1 = (-22x+29)+(-43 + 34x)\partial,\quad P_2=-\tfrac{2312}{299}.
\]
Since $PL$ has polynomial coefficients, so does 
\begin{alignat*}1
  \frac1pP_1L &= qPL - P_2L\\
  &= (-10-165 x+22 x^2) + (201+65 x-34 x^2)\partial\\
  &\quad{} + (-100+109 x-22 x^2)\partial^2
  - (1 + x) (43 - 34 x)\partial^3.
\end{alignat*}
This operator has degree $\deg_x(L) + \deg_x(P_1) - \deg_x(p)=2$, compared to $\deg_x
(L)+\deg_x(P)=3$ achieved with the original desingularizing operator. There is
no left multiple of $L$ of order~3 and degree less than~2. There is also none of order~4,
but there does exist a left multiple of degree~1 and order~5. It can be obtained from~$P$ by
multiplying from the left by an operator $q_0+q_1\partial+q_2\partial^2\in\set Q[x][\partial]$ of order~2 for which the 
coefficients of $p^3 (q_0+q_1\partial+q_2\partial^2) P$ have low degrees modulo~$p^3$: Taking
\[
  q_0=\tfrac{633+64 x-88 x^2}{89401},\quad
  q_1=\tfrac{8(17 x^2+13x-92)}{89401},\quad
  q_2=\tfrac{1}{299}
\]
we have $(q_0+q_1\partial+q_2\partial^2)P=\frac1{p^3}Q_1 + Q_2$, where
\begin{alignat*}1
  Q_1 &= (841+580 x-436 x^2-148 x^3+59 x^4+12 x^5-4 x^6)\\
      &\quad{}+ (1697-528 x -752 x^2+120 x^3+127 x^4-12 x^5\\
      &\qquad{}-4 x^6)\partial+ (x-7)(-9+x+2 x^2)p\, \partial^2 + p^2\partial^3,\\
  Q_2 &= \tfrac{16 (779+374 x)}{89401}-\tfrac{272 (69+34 x)}{89401}\partial.
\end{alignat*}
Set $Q:=\frac1{p^3}Q_1$. Then, since $PL$ has polynomial coefficients, so does
\begin{alignat*}1
QL&=(q_0+q_1\partial+q_2\partial^2)PL-Q_2L\\
&=(2+x)
 +(-3+x)\partial
 -(8+2 x)\partial^2\\
&\quad{}
 +(2-2 x)\partial^3
 +(6+x)\partial^4
 +(1+x)\partial^5.
\end{alignat*}
Its degree is $\deg_x(L) + \deg_x(Q_1) - 3\deg_x(p)=1$. 

As the factor $x+1$ cannot be removed from~$L$, we cannot hope to reduce the
degree even further.  We have thus found that the region of all points
$(r,d)\in\set N^2$ such that there is a left $\set Q(x)[\partial]$-multiple of
$L$ of order $r$ and with polynomial coefficients of degree at most~$d$ is given
by $((2,4)+\set N^2)\cup((3,2)+\set N^2) \cup((5,1)+\set N^2)$.

In Section~\ref{sec:curve} we explain the construction of the operators $Q$ that
turn a desingularizing operator into one that lowers all the degrees as far as
possible, and we give a formula that describes the points $(r,d)$ for which such
a $Q$ exists.

\section{Partial Desingularization}\label{sec:desing}

In this section we discuss under which circumstances an operator~$L$ admits a
left multiple~$PL$ in which a factor of the leading coefficient of $L$ is
removed. This is of interest in its own right, and will also serve as the starting
point for the construction described in the following section. In view of this
latter application, we cover here a slightly generalized variant of
desingularization, which not only applies to the case where a factor can be
completely removed, but also cases where only the multiplicity of the factor can
be lowered.

\begin{example} 
  In the shift case (i.e., $\partial x = (x+1)\partial$), consider the operator
  \begin{alignat*}1
   L &= (3 + x) (9 + 7 x + x^2)- (33 + 70 x + 47 x^2 + 12 x^3 + x^4)\partial \\
     &\quad{}+ (2 + x)^2 (3 + 5 x + x^2)\partial^2 .
  \end{alignat*}
  The factor $(x+2)^2$ in the leading coefficient cannot be removed completely. Yet we can find 
  a multiple in which $x+2$ appears in the leading coefficient (in shifted form) with multiplicity
  one only. One such left multiple of $L$ is 
  \begin{alignat*}1
   &(402+208x+25x^2)
   -(514+743x+258x^2+25x^3)\partial\\
   &\quad{}
   +(233+378x+183x^2+25x^3)\partial^2
   -9(3+x)\partial^3.
  \end{alignat*}
\end{example}

We speak in this case of a partial desingularization. The general definition is
as follows. We formulate it for operators in an arbitrary Ore algebra $\set
O:=A[\partial]:=A[\partial;\sigma, \delta]$ where $A$ is a $\set K$-algebra (in
our case typically $A=\set K[x]$ or $A=\set K(x)$), $\set K$ is a field,
$\sigma\colon A\to A$ is an automorphism and $\delta\colon A\to A$ a
$\sigma$-derivation, i.e., a $\set K$-linear map satisfying the skew Leibniz
rule $\delta(pq) = \delta(p)q + \sigma(p)\delta(q)$ for $p,q\in A$.
For any~$f\in A$, the multiplication rule in~$A[\partial;\sigma, \delta]$ 
is~$\partial f = \sigma(f) \partial + \delta(f)$.
We write $\deg_\partial(L)$ for the \emph{order} of $L\in A[\partial]$, and
if $A=\set K[x]$, we write $\deg_x(L)$ for the maximum degree among the
polynomial coefficients of~$L$. For general information about Ore algebras, see~\cite{bronstein96}. 

\begin{defi}\label{def:remove}
  Let $L\in\set K[x][\partial;\sigma,\delta]$ and let $p\in\set K[x]$ be 
  such that $p\mid\lc_\partial(L)\in\set K[x]$. We say that $p$ is 
  \emph{removable} from $L$ at order~$n$ if there exists some
  $P\in\set K(x)[\partial]$ with $\deg_\partial(P)=n$ 
  and some $w,v\in\set K[x]$ with $\gcd(p,w)=1$ such that
  $PL\in\set K[x][\partial]$ and $\sigma^{-n}(\lc_\partial(PL))=\frac{w}{vp}\lc_\partial(L)$.
  We then call $P$ a $p$-removing operator for~$L$, and $PL$ the corresponding $p$-removed operator.
  $p$~is simply called \emph{removable} from $L$ if it is removable at order~$n$ for some $n\in\set N$.

  If $\gcd(p,\lc_\partial(L)/p)=1$, we say \emph{desingulariz[able|ing|ed]}
  instead of remov$[$able|ing|ed\/$]$, respectively.
\end{defi}

The backwards shift $\sigma^{-n}$ in the definition above is introduced in order
to compensate the effect of the term $\partial^n$ in $P$ on the leading coefficient
on $L$ (i.e., $\lc_\partial(\partial^n L)=\sigma^n(\lc_\partial(L))$.)
Moreover, observe that in this
definition, removing a polynomial $p$ does not necessarily mean that the
$p$-removed operator has no roots of (some shift of) $p$ in its leading
coefficient.  If $L$ contains some factors of higher multiplicity, as in the
example above, then removal of a polynomial is defined so as to respect
multiplicities. Also observe that in the definition we allow that some new
factors $w$ are introduced when~$p$ is removed. This is only a matter of
convenience.  We will see below that we may always assume $v=w=1$, i.e., if
something can be removed at the cost of introducing new factors into the leading
coefficient, then it can also be removed without introducing new factors. The
justification rests on the following lemma.

\begin{lemma}\label{lem:1}
  Let $L\in \set K[x][\partial;\sigma,\delta]$, let $p\in\set K[x]$ with $p\mid\lc_\partial(L)$
  be removable from~$L$, and let $P\in\set K(x)[\partial;\sigma,\delta]$ be a $p$-removing operator for~$L$ with $\deg_\partial(P) = n$.
  \begin{enumerate}
  \item\label{lem:1:2} If $U\in \set K[x][\partial]$ with $\gcd(\lc_\partial(U),\sigma^{n+\deg_\partial(U)}(p))=1$, 
    then $UP$ is also a $p$-removing operator for~$L$.
  \item\label{lem:1:3} If $P=P_1+P_2$ for some $P_1\in\set K(x)[\partial]$ with $\deg_\partial(P_1) = n$ and $P_2\in\set K[x][\partial]$,
    then $P_1$ is also a $p$-removing operator for~$L$.
  \item\label{lem:1:4} There exists a $p$-removing operator $P'$ with $\deg_\partial(P') = n$ and with 
    $p\sigma^{-n}(\lc_\partial(P'L))=\lc_\partial(L)$.
  \end{enumerate}
\end{lemma}
\begin{proof} Let $v,w\in\set K[x]$ be as in Definition~\ref{def:remove}, i.e., 
  $\gcd(p,w)=1$ and $vp\sigma^{-n}(\lc_\partial(PL))=w\lc_\partial(L)$.

  \ref{lem:1:2}. Since $PL$ is an operator with polynomial coefficients, so is~$UPL$.
  Furthermore, with $u=\lc_\partial(U)$ and $m=\deg_\partial(U)$ we have
  \[
    vp\sigma^{-n-m}(\lc_\partial(UPL))=\sigma^{-n-m}(u)w\lc_\partial(L).
  \]
  Since $\gcd(u,\sigma^{n+m}(p))=1$, we have $\gcd(\sigma^{-n-m}(u)w, p)=1$, as required. 

  \ref{lem:1:3}. Clearly, $P_2\in\set K[x][\partial]$ implies $P_2L\in\set
  K[x][\partial]$. Since also $PL\in\set K[x][\partial]$, it follows that
\[ 
  P_1L = (P-P_2)L=PL-P_2L\in\set K[x][\partial].
\]
If $\deg_\partial(P_2)<n$, then we have $\lc_\partial(PL)=\lc_\partial(P_1L)$,
so there is nothing else to show. If $\deg_\partial(P_2)=n$, then
$\lc_\partial(P_1L)=\lc_\partial(PL)-\lc_\partial(P_2L)$ and therefore
\begin{alignat*}1
  vp\sigma^{-n}(\lc_\partial(P_1L))
 &=vp\sigma^{-n}(\lc_\partial(PL)-\lc_\partial(P_2L))\\
 &=(w-vp\sigma^{-n}(\lc_\partial(P_2)))\lc_\partial(L).
\end{alignat*}
Since $\gcd(p, w-vp\sigma^{-n}(\lc_\partial(P_2)))=\gcd(p,w)=1$, the claim follows.

  \ref{lem:1:4}. By the extended Euclidean algorithm we can find $s,t\in\set K[x]$ with
$1 = s w + t p v$. 
Then $\sigma^n(s)P$ is $p$-removing of order~$n$ by part~\ref{lem:1:2}
($\sigma^{n}(s)$~is obviously coprime to~$\sigma^n(p)$), and its leading coefficient is
\[
 \sigma^n\Bigl(\frac{sw}{pv}\Bigr) =\frac1{\sigma^n(pv)} - \sigma^n(t).
\]
By part~\ref{lem:1:3} we may discard the polynomial part~$\sigma^n(t)$, obtaining a $p$-removing operator~$P'$ with
the desired property. 
\end{proof}

The lemma implies that if there is a $p$-removing operator at all, then
there is also one in which all the denominators are powers of $\sigma^n(p)$
(because any factors coprime with $p$ can be cleared according to part~\ref{lem:1:2}), 
and where all numerators have smaller degree than the corresponding denominators 
(because polynomial parts can be removed according to part~\ref{lem:1:3}). 

Similarly as in the proof of part~\ref{lem:1:4}, we can also reduce the problem of removing
a composite polynomial to the problem of removing powers of irreducible polynomials.
For example, if $p=p_1p_2$ is removable from~$L$, where $p_1,p_2\in\set K[x]$ are coprime, 
then obviously both $p_1$ and $p_2$ are removable. Conversely, if $p_1$ and 
$p_2$ are removable, and if $P_1,P_2$ are removing operators of orders~$n_1,n_2$
with $\lc_\partial(P_1)=1/\sigma^{n_1}(p_1)$ and
$\lc_\partial(P_2)=1/\sigma^{n_1}(p_2)$, then for $n=\max\{n_1,n_2\}$ and $u_1,u_2\in\set K[x]$ with
\[
 u_1\sigma^n(p_2)+u_2\sigma^n(p_1)=\gcd(\sigma^n(p_1),\sigma^n(p_2))=1
\] 
the operator $P:=u_1\partial^{n-n_1}P_1+u_2\partial^{n-n_2}P_2\in\set
K(x)[\partial]$ is such that $PL\in\set K[x][\partial]$ and
$\lc_\partial(PL)=\lc_\partial(L)/\sigma^n(p)$.

In summary, in order to determine whether a polynomial
$p=p_1^{k_1}p_2^{k_2}\cdots p_m^{k_m}$ is removable from an operator~$L$, it
suffices to be able to check for an irreducible polynomial $p_i$ and a given
$k_i\geq1$ whether $p_i^{k_i}$ is removable.  Let now $p$ be an irreducible
polynomial and $k\geq1$. If there exists a $p^k$-removing operator, then it can
be assumed to be of the form
\[
  P = \frac{p_0}{\sigma^n(p)^{e_0}} + 
      \frac{p_1}{\sigma^n(p)^{e_1}}\partial + \cdots +
      \frac{p_{n-1}}{\sigma^n(p)^{e_{n-1}}}\partial^{n-1} + 
      \frac{1}{\sigma^n(p)^k}\partial^n,
\]
for some $e_0,\dots,e_{n-1}\in\set N$, and $p_0,\dots,p_{n-1}\in\set K[x]$ with
$\deg_x(p_i)< e_i\deg_x(p)$. In order to decide whether such an operator exists,  
it is now enough to know a bound on $n$ as well as a bound~$e$ on the
exponents~$e_i$, for if $n$ and $e$ are known, we can make an ansatz $p_i =
\sum_{j=0}^{e-1} p_{i,j} x^j$ with undetermined coefficients~$p_{i,j}$, then
calculate $PL$ and rewrite all its coefficients in the form $a/\sigma^n(p)^e +
b$ for some polynomials $a,b$ depending linearly on the undetermined $p_{i,j}$,
then compare the coefficients of the various $a$'s with respect to $x$ to zero
and solve the resulting linearly system for the~$p_{i,j}$.

How the bounds on $n$ and $e$ are derived depends on the particular Ore
algebra at hand. In this paper, we give a complete treatment of the 
shift case ($(\sigma p)(x) = p(x+1)$, $\delta=0$) and make some remarks
about the differential case ($\sigma=\mathrm{id}$, $\delta=\frac{d}{dx}$).
For other cases, see the preprint~\cite{chyzak10}.

\subsection{Shift Case}

In this section, let $\set K[x][\partial]$ denote the Ore algebra of recurrence
operators, i.e., $\sigma$ is the automorphism mapping $x$ to $x+1$ and $\delta$
is the zero map. This case was studied by Abramov and van Hoeij~\cite{abramov99b,abramov06b}. We
give below a new proof of their result, and extend it to the case of partial
desingularization. For consistency with the differential case, we formulate the
result for the leading coefficients, while Abramov and van Hoeij consider the
analogous for the trailing coefficients. Of course, this difference is
immaterial.

We proceed in two steps. First we give a bound on the order of a removing operator (Lemma~\ref{lem:2}),
and then, in a second step, we provide a bound on the exponents in the denominators (Theorem~\ref{thm:1}).
As explained above, it is sufficient to consider the case of removing powers of irreducible 
polynomials, and we restrict to this case. 

\begin{lemma}\label{lem:2}
  Let $L=\ell_0+\ell_1\partial+\cdots+\ell_r\partial^r\in\set K[x][\partial]$ 
  with $\ell_0,\ell_r\neq0$, and let $p$ be an irreducible factor of $\lc_\partial(L)$
  such that $p^k$ is removable from $L$ for some $k\geq1$. 
  Let $n\in\set N$ be s.t.\ $\gcd(\sigma^n(p), \ell_0)\neq 1$ and
  $\gcd(\sigma^m(p), \ell_0)=1$ for all $m>n$. 
  Then $p^k$ is removable at order~$n$ from~$L$.
\end{lemma}
\begin{proof}
  By assumption on~$L$, there exists a $p^k$-removing operator~$P$, say of order~$m$, and by
  the observations following Lemma~\ref{lem:1} we may assume that 
  \[
    P = \frac{p_0}{\sigma^m(p)^{e_0}} + 
        \frac{p_1}{\sigma^m(p)^{e_1}}\partial + \cdots +
        \frac{p_m}{\sigma^m(p)^{e_m}}\partial^m,
  \]
  for $e_i\in\set N$ and $p_i\in\set K[x]$ with $\deg_x(p_i) < e_i\deg_x(p)$ ($i=0,\dots,m$). We may further assume
  $\gcd(\sigma^m(p), p_i)=1$ for $i=0,\dots,m$ (viz.\ that the $e_i$ are chosen minimally). 

  Suppose that $m>n$. We show by induction that then $e_0=e_1=\cdots=e_{m-n-1}=0$,
  so that $p_i=0$ for $i=0,\dots,m-n-1$, i.e., the operator $P$ has in fact the form 
  \[
    P = \frac{p_{m-n}}{\sigma^m(p)^{e_{m-n}}}\partial^{m-n} + \cdots + \frac{p_m}{\sigma^m(p)^{e_m}}\partial^m.
  \]
  Thus $\partial^{n-m}P\in\set K(x)[\partial]$ is a $p^k$-removing operator of order~$n$.

  Consider the operator $T:=\sum_{i=0}^{r+m}t_i\partial^i:= P L \in\set K[x][\partial]$. 
  {}From $\frac{p_0}{\sigma^m(p)^{e_0}}\ell_0=t_0\in\set K[x]$ it follows that $e_0=0$, because
  \[
    \gcd(\sigma^m(p),p_0)=\gcd(\sigma^m(p),\ell_0)=1
  \]
  by the choice of $p_0$ and the assumption in the lemma, respectively, and this leaves no
  possibility for cancellation. 
  
  Assume now, as induction hypothesis, that $e_0=e_1=\dots=e_{i-1}=0$ for some $i<m-n$. Then from
  \begin{alignat*}1
    t_i &= \frac{p_i}{\sigma^m(p)^{e_i}}\sigma^i(\ell_0) + \frac{p_{i-1}}{\sigma^m(p)^{e_{i-1}}}\sigma^{i-1}(\ell_1)
         + \cdots + \frac{p_0}{\sigma^m(p)^{e_0}}\ell_i \\
        &= \frac{p_i}{\sigma^m(p)^{e_i}}\sigma^i(\ell_0)
  \end{alignat*}
  it follows that $\sigma^m(p)^{e_i}\mid p_i\sigma^i(\ell_0)$.
  By the choice of $p_i$ we have $\gcd(\sigma^m(p), p_i)=1$ and by the assumption in the lemma
  we have $\gcd(\sigma^{m-i}(p), \ell_0)=1$ (because $m-i>n$), so it follows that $e_i=0$.
  Inductively, we obtain $e_0=e_1=\cdots=e_{m-n-1}=0$, which completes the proof. 
\end{proof}

It can be shown that $p$ cannot be removed from $L$ if $\sigma^n(p)$ is coprime
with the trailing coefficient of $L$ for all $n\in\set N$ by a variant of~\cite[Lemma 3.]{abramov99b}, so the above lemma covers all
situations where removing of a factor is possible. 

In order to formulate the result about the possible exponents in the
denominator, it is convenient to first introduce some notation. Let us call two irreducible
polynomials $p,q\in\set K[x]\setminus\{0\}$ \emph{equivalent} if
there exists $n\in\set Z$ such that $\sigma^n(p)/q\in\set K$. 
We write $[q]$ for the equivalence class of~$q\in\set K[x]\setminus\{0\}$.
If $p,q$ are equivalent in this sense, we write $p\leq q$ if $\sigma^n(p)/q\in\set K$
for some $n\geq0$, and $p>q$ otherwise. 

The irreducible factors of a polynomial $u\in\set K[x]$ can be grouped into equivalence
classes, for example
\begin{alignat*}1
  u &= \underbrace{(x-4) (x-1)^3 x (x+1)^2} \underbrace{(2x-5)(2x+3)^2(2x+9)}\\
    &\quad\times\underbrace{(x^2+5x+1)(x^2+11x+25)^3}.
\end{alignat*}
For any monic irreducible factor $p$ of $u\in\set K[x]$, let $v_p(u)$ denote the multiplicity of $p$ in~$u$,
and define
\[
  v_{<p}(u):=\max\{\, v_q(u) \mid q \in [p] : p > q\, \}.
\]
For example, for the particular~$u$ above we have $v_{x-4}(u)=1$, $v_{<x-4}(u)=0$, $v_{<x+1}(u)=3$, and so on.

Besides being applicable not only to desingularization but also removal of any factors, 
the following theorem also refines the corresponding result of Abramov and van Hoeij in so far as
their version only covers the case of desingularizing $L$ at some $p$ with $v_{>p}(\lc_\partial(L))=0$
whereas we do not need this assumption. 

\begin{theorem}\label{thm:1}
  Let $L=\ell_0+\ell_1\partial+\cdots+\ell_r\partial^r\in\set K[x][\partial]$ with $\ell_0,\ell_r\neq0$,
  and let $p$ be an irreducible factor of $\ell_r$ such that $p^k$ is removable from $L$ for some $k\geq1$.
  Let $n\in\set N$ be such that $\gcd(\sigma^n(p),\ell_0)\neq1$ and $\gcd(\sigma^m(p),\ell_0)=1$ for
  all $m>n$. Then there exists a $p^k$-removing operator~$P$ for~$L$ and $p$ of the form 
  \[
    P = \frac{p_0}{\sigma^n(p)^{e_0}} + 
        \frac{p_1}{\sigma^n(p)^{e_1}}\partial + \cdots +
        \frac{p_n}{\sigma^n(p)^{e_n}}\partial^n,
  \]
  for some $e_i\in\set N$ and $p_i\in\set K[x]$ with
  \begin{enumerate}
  \item $\deg_x(p_i) < e_i\deg_x(p)$ and $\gcd(\sigma^n(p), p_i)=1$, and
  \item $e_i\leq k + n\,v_{<p}(\lc_\partial(L))$
  \end{enumerate}
  for $i=0,\dots,n-1$, and $p_n=1$, $e_n=k$. 
\end{theorem}
\begin{proof}
  Lemmas \ref{lem:1} and~\ref{lem:2} imply the existence of an operator $P$ with all the required
  properties except possibly the exponent estimate in item~2. Let $P$ be such an operator, 
  and consider the operator $T:=\sum_{i=0}^{r+n}t_i\partial^i:= P L \in\set K[x][\partial]$. 

  Let $e=\max\{e_1,\dots,e_n\}$ and 
  $\bar P:=\sum_{i=0}^n\bar p_i\partial^i:=\sigma^n(p)^e P$.
  Then $\bar p_i = \sigma^n(p)^{e-e_i}p_i$ ($i=0,\dots,n$) 
  and $\sigma^n(p)^e T = \bar P L$
  and $\gcd(\bar p_0,\dots,\bar p_n,\sigma^n(p))=1$.

  Abbreviating $v:=v_{<p}(\lc_\partial(L))$, assume that $e>k+n\,v$. 
  We will show by induction that then $\bar p_i$ contains $\sigma^n(p)$ with
  multiplicity more than $i\,v$ for $i=n,n-1,\dots,0$, which
  is inconsistent with $\gcd(\bar p_0,\dots,\bar p_n,\sigma^n(p))=1$.

  First it is clear that $\bar p_n=\sigma^n(p)^e p_n\sigma^n(\ell_r)$ 
  contains $\sigma^n(p)$ with multiplicity $\geq e-k>nv$, because $P$
  is $p^k$-removing. 
  Suppose now as induction hypothesis that
  there is an $i\geq0$ such that $\sigma^n(p)^{j\,v+1}\mid\bar p_j$ for
  $j=n,n-1,\dots,i+1$. Consider the equality
  \[
   \sigma^n(p)^e t_{i+r} = \bar p_i\sigma^i(\ell_r)
                     + \bar p_{i+1}\sigma^{i+1}(\ell_{r-1})
                     + \cdots
                     + \bar p_n\ell_{r-n},
  \]
  where we use the convention $\ell_j:=0$ for $j<0$.
  The induction hypothesis implies that $\sigma^n(p)^{(i+1)v+1}\mid\bar p_j$ 
  for $j=n,n-1,\dots,i+1$. 
  Furthermore, since $(i+1)v\leq nv<e$, we have $\sigma^n(p)^{(i+1)v+1}\mid \sigma^n(p)^e t_{i+r}$. 
  Both facts together imply $\sigma^n(p)^{(i+1)v+1}\mid \bar p_i\sigma^i(\ell_r)$.
  The definition of~$v$ ensures that $\sigma^n(p)$ is contained in $\sigma^i(\ell_r)$ with 
  multiplicity at most~$v$, so it must be contained in $\bar p_i$ with multiplicity
  more than $(i+1)v-v=i\,v$, as claimed. 
\end{proof}

\subsection{Differential Case}

In this section $\set K[x][\partial]$ refers to the Ore algebra of differential
operators, i.e., $\sigma = \operatorname{id}$ and $\delta = \frac{d}{dx}$. Let $L \in
\set K[x][\partial]$ and suppose for simplicity that $p=x$ is a factor of $\lc_\partial(L)$.
In~\cite{abramov06b}, the authors show that $L$ can be desingularized at $x$ if and only
if $x=0$ is an \emph{apparent singularity,} that is, if and only if $L(y) = 0$
admits $\deg_\partial(L)$ linearly independent formal power series solutions.  The
authors furthermore give an algorithm to find an operator~$P$ such that
if $\xi$ is either an ordinary point of $L$ or an apparent singularity of~$L$, 
then $\xi$ is an ordinary point of $PL$. Therefore this
algorithm desingularizes all the points that can be desingularized.  The authors
also give a sharp bound for $\deg_\partial(P)$.  The authors furhtermore give
some indications concerning partial desingularizations.  It would be interesting
to give a complete algorithm for partial desingularizations.

\section{Order-Degree Curves}\label{sec:curve}

We now turn to the construction of left multiples of $L$ with polynomial
coefficients of small degree, and to the question of how small these degrees can be
made. As already indicated in Section~\ref{sec:2}, we start from an operator~$P$
which removes some factor from the leading coefficient of~$L$, say it removes a
polynomial~$p$ of degree~$k$. According to Lemma~\ref{lem:1}, we may assume that
$\lc_\partial(P)=1/\sigma^{\deg_\partial(P)}(p)$ and that all other coefficients of $P$ are rational
functions whose numerators have lower degree than the corresponding
denominators.  Thus we already have $\deg_x(PL)\leq\deg_x(L)-1$. Furthermore, if
$q$ is any polynomial with $\deg_x(q)<\deg_x(p)=k$, then multiplying $P$ by~$q$
(from left) and removing polynomial parts by Lemma~\ref{lem:1}.\ref{lem:1:3} gives another operator~$Q$ with
$\deg_x(QL)\leq\deg_x(L)-1$.  All the operators $Q$ obtained in this way form a
$\set K$-vector space of dimension~$k$. Within this vector space we search for
elements where $\deg_x(QL)$ is as small as possible. Forcing the coefficients of
the highest degrees to zero gives a certain number of linear constraints which
can be balanced with the number of degrees of freedom offered by the
coefficients of~$q$, as illustrated in the figure below.
As long as we force fewer than $k$ terms to zero, we will find a nontrivial solution.

\begin{center}
\begin{tikzpicture}[scale=.9]
  \foreach \x in {0, .3} 
     \foreach \y in {0, .3, ..., 2.2} 
        \fill (\x, \y) circle (2pt);  
  \fill (.6, 0) circle (2pt) (.6, .3) circle (2pt);
  \foreach \y in {.6, .9, ..., 1.5}
    \draw (.6, \y) circle (2pt);
  \draw[rounded corners=4pt] (-.15, 1.65) rectangle (.45, 2.25) (.45, .45) rectangle (.75, 1.65);
  \draw[->] (.45, 1.95) .. controls (.6, 1.95) .. (.6, 1.7);
\end{tikzpicture}
\end{center}

If we want to eliminate $k$ terms or more in order to get a result of even lower
degree, we need more variables. We can create $k$ more variables if instead of
an ansatz $qP$ we make an ansatz $(q_0+q_1\partial)P$ for some $q_0,q_1\in\set
K[x]$ with $\deg_x(q_0),\deg_x(q_1)<k$.  Again removing all polynomial parts
from the rational function coefficients we obtain a vector
space of operators $Q$ with $\deg_x(QL)\leq\deg_x(L)-1$ whose dimension
is~$2k$. The additional degrees of freedom can be used to eliminate more high
degree terms, the result being an operator of lower degree but higher order. If we
let the order increase further and for each fixed order use all the available
degrees of freedom to reduce the degrees to minimize the degrees of the
polynomial coefficients, a hyperbolic relationship between the order and the
degree of $QL$ emerges.
In Theorem~\ref{thm:curve} below, we make this relationship precise, taking into account that
for a given operator~$L$ the leading coefficient may contain several factors~$p$ that
are removable at different orders~$n$. The resulting region of all
points $(r,d)\in\set N^2$ for which there exists a left multiple of~$L$ of order~$r$ with
polynomial coefficients of degree at most~$d$ is then given by an overlay of a finite number
of hyperbolas. 

Before turning to the proof of this theorem, let us illustrate its basic idea with the example
operators from Section~\ref{sec:2}.

\begin{example}
  Let $L\in\set Q[x][\partial]$, $p\in\set Q[x]$, and $P\in\set Q(x)[\partial]$ be as in Section~\ref{sec:2}. 
  Recall that $p$ is an irreducible cubic factor of $\lc_\partial(L)$ and that $P$ is a $p$-removing operator for~$L$.
  We have $P=\frac{p_1}{p}\partial + \frac{p_0}{p}$ for some $p_1,p_0\in\set Q[x]$ 
  with $\deg_x(p_1)=0$ and $\deg_x(p_0)=2$. 
  We have seen in Section~\ref{sec:2} that there is an operator $Q\in\set Q(x)[\partial]$ of order~3
  such that $QL\in\set Q[x][\partial]$ and $\deg_x(QL)=1$. 
  Our goal here is to explain why this operator exists. 

  Make an ansatz $Q_1=(q_0+q_1\partial+q_2\partial^2)P$ with undetermined polynomials $q_0,q_1,q_2$.
  After expanding the product and applying commutation rules, $Q_1$~has the form 
  \begin{alignat*}1
    &\frac{p_1q_2}{p}\partial^3
    + \frac{(\ldots)q_2+(\ldots)q_1}{p^2}\partial^2\\
    &{}+ \frac{(\ldots)q_2+(\ldots)q_1+(\ldots)q_0}{p^3}\partial
    + \frac{(\ldots)q_2+(\ldots)q_1+(\ldots)q_0}{p^3},
  \end{alignat*}
  where the $(\ldots)$ are certain polynomials whose precise form is irrelevant for our purpose. 

  Note that by Lemma~\ref{lem:1}.\ref{lem:1:2}, $Q_1L\in\set Q[x][\partial]$ regardless of the choice of
  $q_0,q_1,q_2$, and that by Lemma~\ref{lem:1}.\ref{lem:1:3}, this property is not lost if we add
  to $Q_1$ some operator in $\set Q[x][\partial]$ of our choice. 
  Therefore, if $Q_2\in\set Q[x][\partial]$ is the operator obtained from $p^3 Q_1\in\set Q[x][\partial]$ 
  by reducing all the coefficients modulo~$p^3$, then $p^{-3}Q_2L\in\set Q[x][\partial]$, still regardless
  of the choice of $q_0,q_1,q_2$. 

  The coefficients of $Q_2$ depend linearly on the undetermined polynomials $q_0,q_1,q_2$. 
  If we choose their degree to be $\deg_x(p)-1=2$, then we have $3(2+1)=9$ variables for the coefficients
  of $q_0,q_1,q_2$. Choosing a higher degree would give more variables but also introduce undesired 
  solutions such as $q_0=q_1=q_2=p$, for which the reduction modulo $p^3$ leads to the useless result $Q_2=0$.
  This cannot happen if we enforce $\deg_x(q_i)<\deg_x(p)$.

  The operator $p^{-3}Q_2L$ has degree 
  \[
   \deg_x(Q_2)+\deg_x(L)-3\deg_x(p)=\deg_x(Q_2)-5,
  \]
  which is equal to $1$ if $\deg_x(Q_2)=6$.
  A priori, the degree of $Q_2$ in~$x$ may be up to $\deg_x(p^3)-1=8$.
  In order to bring it down to~6, we equate the coefficients of $x^i\partial^j$ for $i=7,8$ and
  $j=0,\dots,3$ to zero. This gives 8~equations. As there are more variables than equations, there
  must be a nontrivial solution. 

\end{example}

For formulating the proof of the general statement, it is convenient to work with an alternative
formulation of removability, which is provided in the following lemma.
Throughout the section, $\set K[x][\partial]=\set K[x][\partial;\sigma,\delta]$
is an arbitrary Ore algebra.

\begin{lemma}\label{lem:3}
  $p\in\set K[x]$ is removable from $L\in\set K[x][\partial]$ at order~$n$ 
  if and only if 
  there exists $P\in\set K[x][\partial]$ with $\deg_\partial(P)=n$ and $PL\in\sigma^n(p)\lc_{\partial}(P)\set K[x][\partial]$.
\end{lemma}
\begin{proof}
  ``$\Leftarrow$'': $P_0=\frac{1}{\sigma^n(p)\lc_{\partial}(P)}P$ is a $p$-removing operator.   

  ``$\Rightarrow$'': 
  Start from a $p$-removing operator of the form
  \[
    P_0 = \sum_{i=0}^{n-1} \frac{p_i}{\sigma^n(p)^{e_i}}\partial^i + \frac1{\sigma^n(p)}\partial^n,
  \]
  and set $P=\sigma^n(p)^e P_0$ where $e=\max\{e_0,\dots,e_{n-1},1\}\geq 1$. 
  Because of $P_0L\in\set K[x][\partial]$ it follows that 
  \begin{alignat*}1
    PL&\in\sigma^n(p)^e\set K[x][\partial]=\sigma^n(p)\lc_\partial(P)\set K[x][\partial].
    \rlap{\quad\qed} % hack avoiding lonley line containing only the qed box
  \end{alignat*}\def\qed{} % hack avoiding lonley line containing only the qed box
\end{proof}

The next lemma is a generalization of Bezout's relation to more than two coprime polynomials,
which we will also need in the proof. 

\begin{lemma}\label{lem:4}
  Let $u_1,\dots,u_m\in\set K[x]$ be pairwise coprime and $u=u_1u_2\cdots u_m$,
  and let $v_1,\dots,v_m\in\set K[x]$ be such that $\deg_x(v_i)<\deg_x(u_i)$ ($i=1,\dots,m$). 
  If 
  \[
    \sum_{i=1}^m v_i\frac{u}{u_i}=0
  \]
  then $v_1=v_2=\cdots=v_m=0$. 
\end{lemma}
\begin{proof}
  Since the $u_i$ are pairwise coprime, $u_i\nmid u/u_i$ for all~$i$.
  However, $u_i\mid u/u_j$ for all $j\neq i$. Both facts together with 
  $\sum_{i=1}^m v_i u/u_i = 0$ imply that $u_i\mid v_i$ for all~$i$.
  Since $\deg_x(v_i)<\deg_x(u_i)$, the claim follows. 
\end{proof}

\begin{theorem}\label{thm:curve}
  Let $L\in\set K[x][\partial]$, and let $p_1,\dots,p_m\in\set K[x]$ be 
  factors of $\lc_\partial(L)$ which are removable at orders $n_1,\dots,n_m$, respectively,
  so that the $\sigma^{n_i}(p_i)$ are pairwise coprime. 
  Let $r\geq\deg_\partial(L)$ and
  \[
    d\geq\deg_x(L) - \biggl\lceil\sum_{i=1}^m\Bigl(1-\frac{n_i}{r-\deg_\partial(L)+1}\Bigr)^{\!+}\deg_x(p_i)\biggr\rceil,
  \]
  where we use the notation $(x)^+:=\max\{x,0\}$.
  Then there exists an operator $Q\in\set K(x)[\partial]\setminus\{0\}$ 
  such that $QL\in\set K[x][\partial]$ and $\deg_\partial(QL)=r$ and $\deg_x(QL)=d$.
\end{theorem}
\begin{proof}
  Let $r\geq\deg_\partial(L)$, and set $s:=r-\deg_\partial(L)$ so that $s=\deg_\partial(Q)$.
  We may assume without loss of generality that $s$ is such that 
  $1-\frac{n_i}{r-\deg_\partial(L)+1}=1-\frac{n_i}{s+1}>0$ for all $i$
  by simply removing all the $p_i$ for which $1-\frac{n_i}{s+1}\leq 0$ from consideration.
  We thus have $s\geq n_i$ for all~$i$.

  Lemma~\ref{lem:3} yields operators $P_i\in\set K[x][\partial]$
  of order~$n_i$ with $P_iL\in\sigma^{n_i}(p_i)\lc(P_i)\set K[x][\partial]$. 
  Set 
  \[ 
    q=\prod_{i=1}^m\prod_{j=0}^{s-n_i}\sigma^{j+n_i}(p_i)\sigma^j(l_i),
  \]
  where $l_i=\lc_\partial(P_i)$. Consider the ansatz
  \[
    Q_1 = \sum_{i=1}^m\sum_{j=0}^{s-n_i} q_{i,j} \frac{q}{\sigma^{j+n_i}(p_i)\sigma^j(l_i)}\partial^j P_i
  \]
  for undetermined polynomial coefficients $q_{i,j}$ ($i=1,\dots,m$; $j=0,\dots,n_i$) of degree less than $\deg_x(p_i)$.
  Regardless of the choice of these coefficients, we will always have $Q_1\in\set K[x][\partial]$
  and $Q_1L\in q\set K[x][\partial]$. Also, for arbitrary $R\in\set K[x][\partial]$ and 
  $Q_2=Q_1-qR$ we have $Q_2\in\set K[x][\partial]$ and $Q_2L\in q\set K[x][\partial]$. This
  means that we can replace the coefficients in $Q_1$ by their remainders upon division by $q$ 
  without violating any of the mentioned properties of~$Q_1$. 

  Also observe that any operator $Q_2$ obtained in this way is nonzero unless
  all the $q_{i,j}$ are zero, because if $k$ is maximal such that at least 
  one of the $q_{i,k}$ is nonzero, then
  \[
   \lc_\partial(Q_1)=\sum_{i=1}^m q_{i,k}\frac{q}{\sigma^{k+n_i}(p_i)\sigma^k(l_i)}\sigma^k(l_i)
      =\sum_{i=1}^m q_{i,k}\frac{q}{\sigma^{k+n_i}(p_i)}
  \]
  is nonzero by Lemma~\ref{lem:4}. Furthermore, $\lc_\partial(Q_1)\not\equiv0\bmod q$
  because $\deg_x(q_{i,k})<\deg_x(p_i)$ implies 
  $\deg_x(\lc_\partial(Q_1))<\deg_x(q)$. 
  
  The ansatz for the $q_{i,j}$ gives $\sum_{i=1}^m(s-n_i+1)\deg_x(p_i)$
  variables. Plug this ansatz into $Q_1$ and reduce all the polynomial coefficients modulo~$q$,
  obtaining an operator $Q_2$ of degree less than $\deg_x(q)=\sum_{i=1}^m(s-n_i+1)(\deg_x(p_i)+\deg_x(l_i))$.
  Then for each of the $s+1$ polynomial coefficients in $Q_2$ equate the 
  coefficients of the terms $x^j$ for
  \[
    j > \sum_{i=1}^m(s-n_i)\bigl(\deg_x(p_i)+\deg_x(l_i)\bigr)
      + \biggl\lfloor\frac{\sum_{i=1}^m n_i\deg_x(p_i)}{s+1}\biggr\rfloor
  \]
  to zero. This gives altogether 
  \begin{alignat*}1
    &(s+1)\biggl(\sum_{i=1}^m\bigl( (s{-}n_i{+}1)\bigl(\deg_x(p_i){+}\deg_x(l_i)\bigr) - 1 - \deg_x(l_i)\bigr)\\
    &- \sum_{i=1}^m (s{-}n_i)\bigl(\deg_x(p_i){+}\deg_x(l_i)\bigr) 
             - \biggl\lfloor\frac{\sum_{i=1}^m n_i\deg_x(p_i)}{s+1}\biggr\rfloor\biggr)\\
    &=(s+1)\biggl(\sum_{i=1}^m \deg_x(p_i) - 1 - \biggl\lfloor\frac{\sum_{i=1}^m n_i\deg_x(p_i)}{s+1}\biggr\rfloor\biggr)
  \end{alignat*}
  equations. The resulting linear system has a nontrivial solution because
  \begin{alignat*}1
       &\mathrm{\#vars}-\mathrm{\#eqns} \\
    ={}& \sum_{i=1}^m (s-n_i+1)\deg_x(p_i)\\
       &\quad{}
       - (s+1)\biggl(\sum_{i=1}^m\deg_x(p_i) - 1 - \biggl\lfloor\frac{\sum_{i=1}^m n_i\deg_x(p_i)}{s+1}\biggr\rfloor\biggr) \\
    ={}& - \sum_{i=1}^m n_i\deg_x(p_i) - (s+1)\biggl(-1-\biggl\lfloor\frac{\sum_{i=1}^m n_i\deg_x(p_i)}{s+1}\biggr\rfloor\biggr)\\
    >{}& - \sum_{i=1}^m n_i\deg_x(p_i) + \frac{s+1}{s+1}\sum_{i=1}^m n_i\deg_x(p_i) = 0.
  \end{alignat*}
  By construction, the solution gives rise to an operator~$Q_2\in\set K[x][\partial]$ of
  order at most $n$ with polynomial coefficients of degree at most 
  \[
    \sum_{i=1}^m (s-n_i)(\deg_x(p_i) + \deg_x(l_i)) + \biggl\lfloor\frac{\sum_{i=1}^m n_i\deg_x(p_i)}{s+1}\biggr\rfloor,
  \]
  for which $Q_2L\in q\set K[x][\partial]$. Thus if we set $Q=\frac1{q}Q_2\in\set K(x)[\partial]$,
  we have $\deg_\partial(QL)=\deg_\partial(L)+s=r$ and $\deg_x(QL)$ is at most
  \begin{alignat*}1
    & \deg_x(L) + \deg_x(Q_2)-\deg_x(q) \\
    &\leq \deg_x(L) + 
    \sum_{i=1}^m (s-n_i)(\deg_x(p_i) + \deg_x(l_i))\\
    &\quad{} + \biggl\lfloor\frac{\sum_{i=1}^m n_i\deg_x(p_i)}{s+1}\biggr\rfloor\\
    &\qquad{} - \sum_{i=1}^m (s-n_i+1)(\deg_x(p_i)+\deg_x(l_i)) \\
    &\leq \deg_x(L) - \sum_{i=1}^m \deg_x(p_i) + \biggl\lfloor\frac{\sum_{i=1}^m n_i\deg_x(p_i)}{s+1}\biggr\rfloor\\
    &=\deg_x(L) - \biggl\lceil\sum_{i=1}^m \Bigl(1{-}\frac {n_i}{s+1}\Bigr)\deg_x(p_i)\biggr\rceil, 
  \end{alignat*}
  as required. (The final step uses the facts $\lfloor -x\rfloor=-\lceil x\rceil$ 
  and $\lceil x+n\rceil=\lceil x\rceil+n$ for $x\in\set R$ and $n\in\set Z$.)
\end{proof}

\begin{example}
  \begin{enumerate}
  \item 
  Consider again the example from Section~\ref{sec:2}.
  There we started from an operator $L\in\set K[x][\partial]$ of order~2
  and degree~4
  for which there exists a desingularizing operator~$P$ of order~1
  which removes a polynomial~$p$ of degree~3.
  According to the theorem, for every $r\geq2$ exists an operator 
  $Q\in\set K[x][\partial]$ with $QL\in\set K[x][\partial]$,
  $\deg_\partial(QL) \leq r$ and
  \begin{alignat*}1
    d:=\deg_x(QL) &\leq 4 - \Bigl(1 - \frac1{r-2+1}\Bigr)^{\!+}3= \frac{r+2}{r-1}.
  \end{alignat*}
  This hyperbola precisely predicts the order-degree pairs we found in 
  Section~\ref{sec:2}:
  \begin{center}
    \begin{tabular}{c|ccccccccccc}
      $r$ & 2 & 3 & 4 & 5 & 6 & 7 & 8 & 9 \\\hline
      $d$ & 4 & 2 & 2 & 1 & 1 & 1 & 1 & 1
    \end{tabular}
  \end{center}
  \item 
  Consider the sequence $(a_n)_{n=0}^\infty$ defined by 
  \[
    a_n=\sum_k \frac{\Gamma(2n+k)\Gamma(n-k+2)}{\Gamma(2n-k)\Gamma(n+2k)}\quad(n\in\set N).
  \]
  Zeilberger's algorithm finds an annihilating operator~$L$ of the form
  \begin{alignat*}1
    &9(1+n)(1+3n)(2+3n)^2(3n+4)p(n+1)\\
    &\quad{}+ (\text{\dots degree~16\dots})\partial+ (\text{\dots degree~15\dots})\partial^2\\
    &\qquad{}- 10 n(8+5n)(9+5n)(11+5n)(12+5n)p(n)\partial^3,
  \end{alignat*}
  where $\partial$ represents the shift operator and $p$~is a certain irreducible polynomial of degree~10.
  This polynomial is removable of order~1. Therefore, by the theorem, we expect left multiples 
  of~$L$ of order $r$ and degree bounded by
  \[
    16 - \Bigl(1 - \frac1{r-3+1}\Bigr)^{\!+}10 = \frac{6r-2}{r-2}. 
  \]
  In the figure below, the curve $d=\frac{6r-2}{r-2}$ (solid) is contrasted with
  the estimate $d=\frac{8r-1}{r-2}$ (dashed) derived last year for this example~\cite{chen12c}
  as well as the region of all points $(r,d)$ for which a left multiple of $L$ 
  of order~$r$ and degree~$d$ exists (gray). 
  The new curve matches precisely the boundary of the gray region, even including the 
  very last degree drop (which is not clearly visible on the figure): for $r=12$
  we have $\frac{6r-2}{r-2}=7$ and for $r=13$ we have $\frac{6r-2}{r-2}\approx6.9<7$.

% Show[Graphics[{Gray, Rectangle[#, {20,40}]& /@ {{3, 16}, {4, 11}, {5, 9}, {6, 8}, {7, 8}, {8, 7}, {9, 7}, {10, 7}, {11, 7}, {12, 7}, {13, 6}, {14, 6}, {15, 6}, {16, 6}, {17, 6}, {18, 6}, {19, 6}, {20, 6}}, Black, Thickness[.01], Line[Table[{r+3,(6r+16)/(r+1)}, {r, -.9, 20-3, .1}]], Dashed, Line[Table[{r,(8r-1)/(r-2)}, {r, 2.1, 20, .1}]], White, Rectangle[{0,40},{20,50}]}], PlotRange -> {{0,20},{0,40}}, Axes -> True, AxesOrigin -> {0,0}, AspectRatio -> 1, ImageSize -> 165, AxesLabel -> {r, d}]

  \medskip
  \centerline{\qquad\includegraphics[width=.5\hsize]{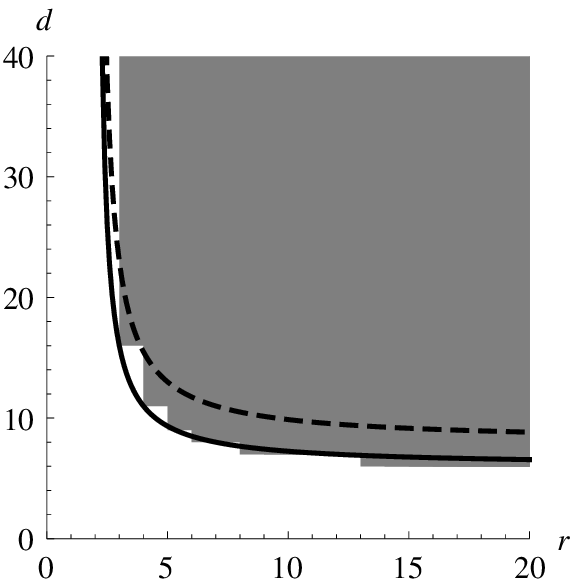}}

\item Consider the minimal order telescoper~$L$ for the hyperexponential
  term in Example~15.2 in~\cite{chen12b}.  It has order~3 and degree~40.  The
  leading coefficient contains an irreducible polynomial $p$ of degree~23 at
  order~1 and otherwise only non-removable factors.  Theorem~\ref{thm:curve}
  therefore predicts left multiples of $L$ of degree $r$ and degree
    \[
      40 - \Bigl(1-\frac1{r-3+1}\Bigr)^{\!+}23 = \frac{17r-11}{r-2}
    \]
    for all $r\in\set N$. Again, this estimate is accurate, while the estimate
    $\frac{24r-9}{r-2}$ derived in~\cite{chen12b}  overshoots.

    \medskip
    \centerline{\qquad\includegraphics[width=.5\hsize]{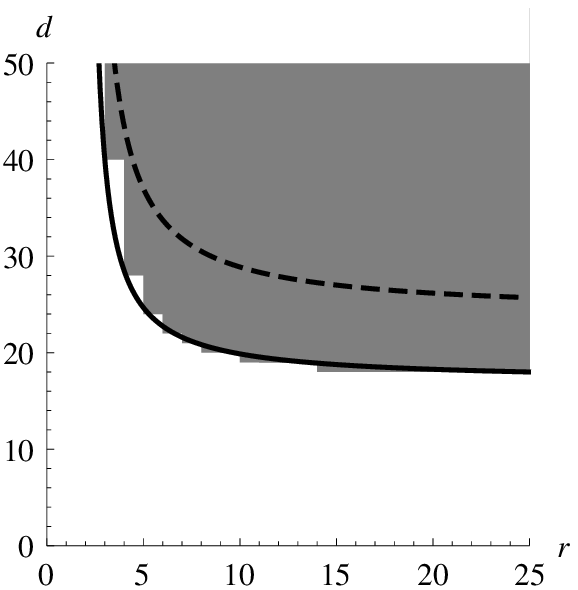}}    

% Show[Graphics[{Gray, Rectangle[#, {20,60}]& /@ {{3, 40}, {4, 28}, {5, 24}, {6, 22}, {7, 21}, {8, 20}, {9, 20}, {10, 19}, {11, 19}, {12, 19}, {13, 19}, {14, 18}, {15, 18}, {16, 18}, {17, 18}, {18, 18}}, Black, Thickness[.01], Line[Table[{r+3,3+(14r+37)/(r+1)}, {r, -.75, 17, .1}]], Dashed, Line[Table[{r, (-9+24r)/(r-2)}, {r, 2.25, 20, .1}]], White, Rectangle[{0,60},{20,70}], White, Rectangle[{0,40},{20,50}]}, AxesOrigin -> {0,0},Axes -> True,PlotRange -> {{0,20},{0,40}}, AspectRatio -> 1, ImageSize -> 165, AxesLabel -> {r,d}]]

  \item Operators coming from applications tend to have leading coefficients that contain a single
    irreducible polynomial of large degree which can be removed at order~1, besides factors that
    are not removable. But Theorem~\ref{thm:curve} also covers the more general situation of 
    factors that are only removable of higher order, and even the case of several polynomials
    that are removable at several orders. As an example for this general situation, consider
    the operator
    \begin{alignat*}1
      L &= 8(1{+}x)(1{+}2x)^3(37{+}3z)^7 (14{+}32x{+}26x^2{+}7x^3)^7 \\
        &\quad{}-9(1{+}3x)^9(2{+}3x)^2(1{+}x{+}5x^2{+}7x^3)^7\partial,
    \end{alignat*}
    where $\partial$ represents the shift operator. From its leading
    coefficient, the polynomial $(1 + x + 5x^2 + 7x^3)^7$ is removable at
    order~1, and in addition, $(1+3x)^7$ is removable at order~12. The remaining
    factors are not removable. According to Theorem~\ref{thm:curve} we expect
    that $L$ admits left multiples of order~$r$ and degree
    \[
      32 - 21\Bigl(1-\frac1{r}\Bigr)^{\!+} - 7\Bigl(1-\frac3{1}\Bigr)^{\!+},
    \]
    for all $r\in\set N$. It turns out that this prediction is again accurate for every~$r$.
    Observe that in this example the curve is a superposition of two hyperbolas. 

    \medskip
    \centerline{\qquad\includegraphics[width=.5\hsize]{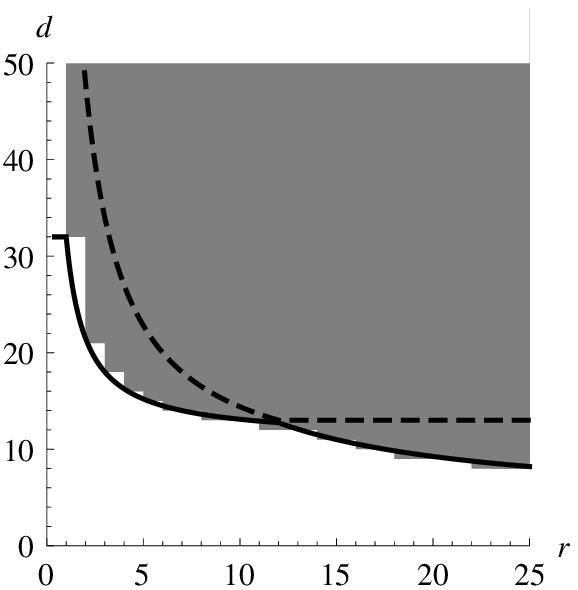}}    

% Show[Graphics[{Gray, Rectangle[#, {25,60}]& /@ {{1, 32}, {2, 21}, {3, 18}, {4, 16}, {5, 15}, {6, 14}, {7, 14}, {8, 13}, {9, 13}, {10, 13}, {11, 12}, {12, 12}, {13, 12}, {14, 11}, {15, 11}, {16, 10}, {17, 10}, {18, 9}, {19, 9}, {20, 9}, {21, 9}, {22, 8}, {23, 8}, {24, 8}, {25, 8}}, Black, Thickness[.01], Line[Table[{r,32 - Max[1 - 1/r, 0]21 - Max[1 - 12/r, 0]7}, {r, .4, 25, .1}]], Dashed, Line[Table[{r,32 - 19 - Min[1 - 1/r, 0]21 - Min[1 - 12/r, 0]7}, {r, .4, 25, .1}]], White, Rectangle[{0,60},{25,70}], White, Rectangle[{0,50},{25,60}]}, AxesOrigin -> {0,0},Axes -> True,PlotRange -> {{0,25},{0,50}}, AspectRatio -> 1, ImageSize -> 165, AxesLabel -> {r,d}]]

  \end{enumerate}
\end{example}

In conclusion, we believe that removable factors provide a universal explanation
for all the order-degree curves that have been observed in recent years for
various different contexts.  We have derived a formula for the boundary of the
gray region associated to a fixed operator~$L$, which, although formally only a
bound, happens to be exact in all the examples we considered.  This does not
immediately imply better complexity estimates or faster variants of algorithms
exploiting the phenomenon of order-degree curves, because usually $L$ is not
known in advance but rather the desired output of a calculation, and therefore
we usually have no information about the removable factors
of~$\lc_\partial(L)$.  However, we now know what we have to look at: in order to
improve algorithms based on trading order for degree, we need to develop a
theory which provides a priori information about the removable factors
of~$\lc_\partial(L)$.  In other words, our result reduces the task of better
understanding order-degree curves to the task of better understanding what
causes the appearance of removable factors in operators coming from
applications.

\vfill\pagebreak

%\bibliographystyle{plain}
%\bibliography{curve}

\end{document}